\documentclass[graybox]{svmult}
\usepackage{mathptmx}       
\usepackage{helvet}         
\usepackage{courier}        
\usepackage{makeidx}         
\usepackage{graphicx}        
\usepackage{multicol}        
\usepackage[bottom]{footmisc}
\usepackage{amsmath,amssymb,amsfonts}        
\makeindex             
\usepackage{hyperref}
\hypersetup{
    colorlinks=true,
    linkcolor=black,
    filecolor=black,      
    urlcolor=black,
    citecolor=black,
    }

\begin{document}
\title*{A note on ``Higher order linear differential equations for unitary matrix integrals: applications and generalisations"}
\author{Peter J. Forrester and Fei Wei}
\institute{Peter J. Forrester \at School of Mathematics and Statistics,  The University of Melbourne,
Victoria 3010, Australia. \email{pjforr@unimelb.edu.au}
\and Fei Wei \at Department of Mathematics, University of Sussex, Brighton, BN1 9RH, United Kingdom.  \email{weif0831@gmail.com}}
%
%
\maketitle

\abstract{
In this note, we briefly introduce the background and motivation of the collaborative work \href{https://arxiv.org/pdf/2508.20797}{[arXiv:2508.20797]}, and provide an outline of the main results.
The latter relates to matrix and higher order scalar differential equations satisfied by certain Hankel and Toeplitz determinants involving I-Bessel functions, or equivalently certain unitary matrix integrals, and moreover puts this property in a broader context. We also investigate large gaps between zeros of the derivatives of the Hardy $\mathsf{Z}$-function, assuming the validity of a certain joint moments conjecture in random matrix theory.
}

\bigskip

The study of random matrices was first observed to have deep connections with number theory in the 1970s. In particular, Montgomery \cite{M73} investigated the statistical distribution of normalized pairs of zeros of the Riemann zeta function on the critical line 
(Riemann zeros)
under the assumption of the Riemann hypothesis. 
 We recall that for a real valued sequence $\{x_n\}_{n=1}^\infty$ the pair correlation, $r_N(s)$ say, is specified as
 ${1 \over N} \# \{ 1 \le m,n\le N, m \ne n \, : \,
 |x_m - x_n | \le s \}$; for the modification of this definition required for the Riemann zeros, see e.g.~\cite{Wi}.
The pair correlation function conjectured in \cite{M73}, now known as Montgomery’s pair correlation conjecture, coincides precisely with the  pair correlation function derived 
by Dyson \cite{D62} for suitably scaled eigenvalues of random matrices from the Circular Unitary Ensemble (CUE) or the Gaussian Unitary Ensemble (GUE). Moreover, in \cite{M73}, Montgomery proved that this conjecture holds for test functions whose Fourier transforms are supported on restricted intervals. However, in its full generality, this conjecture remains open.
Compelling evidence in support of Montgomery's conjecture comes from Odlyzko’s numerical computations of the zeros of the Riemann zeta function at heights near zero number $10^{20}$ \cite{O89},
extended to even higher heights in
\cite{Od01}. His computations of the pair correlation and nearest neighbor spacing of the zeros of the Riemann zeta function agree well with those for the GUE. 

Another close connection between number theory and random matrix theory arises in the study of the moments of the Riemann zeta function on the critical line. Number theorists are particularly interested in the asymptotic formulas for these moments, which are important both for estimating the maximal order of the zeta function and for applications to the distribution of prime numbers, through zero density estimates, as well as to divisor problems. More specifically, one seeks to determine $g_l$ such that
\begin{equation}\label{asyzeta}
\frac{1}{T}\int_0^T \left|\zeta\left(\frac{1}{2}+it\right)\right|^{2l} dt \sim c_l 
\frac{g_l}{\Gamma(1+l^2)} (\log T)^{l^2}, \quad T \to \infty,
\end{equation}  
where $c_l$ is the arithmetic factor given in terms of the particular product over primes $p$ by
\begin{equation}\label{arithmeticfactor}
c_l = \prod_{p} \left(1 - \frac{1}{p} \right)^{l^2} \sum_{m=0}^\infty \left( \frac{\Gamma(m+l)}{m! \, \Gamma(l)} \right)^2 p^{-m}.
\end{equation}

The study of \eqref{asyzeta} dates back to Hardy and Littlewood \cite{HL18}, who investigated the mean square of $\zeta(\frac{1}{2}+it)$, thereby introducing the second moment and establishing $g_1=1$. For $l=2$, Ingham \cite{I26} showed that $g_2=2$. Using the approximate functional equation and Dirichlet polynomial techniques, Conrey and Ghosh \cite{CG98} conjectured $g_3=42$, and Conrey and Gonek \cite{CG01} conjectured $g_4=24024$. For $l>4$, no reliable conjecture exists in number theory, even assuming the Riemann hypothesis. Keating and Snaith \cite{KS00} studied the moments of characteristic polynomials from CUE and established the following asymptotics:
\begin{equation}\label{asyforcue}
\int_{U(N)} |\Lambda_N(1)|^{2l} \, d\mu_{N} \sim 
\prod_{j=0}^{l-1} \frac{j!}{(j+l)!} \, N^{l^2}, \quad N\to \infty,
\end{equation}
where $U(N)$ is the group of unitary matrices, $d\mu_{N}$ is the Haar measure on $U(N)$, and $\Lambda_N(s)$ is the characteristic polynomial defined by
$
\Lambda_N(s) = \prod_{n=1}^N (1 - s e^{-i\theta_n})$,
with $e^{i\theta_1}, \ldots, e^{i\theta_N}$ being the eigenvalues of a Haar-distributed unitary matrix of size $N \times N$. Moreover, Keating and Snaith used the value distribution of this characteristic polynomial  to model that of the Riemann zeta function on the critical line and conjectured that the leading coefficient in \eqref{asyforcue} predicts that of \eqref{asyzeta} in the following way. 

\begin{conjecture}[Keating-Snaith]\label{KSconjecture}
Let $g_l$ be as given in \eqref{asyzeta}. Then
\begin{equation}\label{predicttheleadingcoefficient}
\prod_{j=0}^{l-1} \frac{j!}{(j+l)!} = \frac{g_l}{\Gamma(1+l^2)}.
\end{equation}
\end{conjecture}

This conjecture for $l=1,2,3,4$ coincides with the results mentioned earlier from number theory. Furthermore, the work of Gonek, Hughes, and Keating \cite{GHK07} provides additional evidence supporting this conjecture, and independent importance of the leading coefficients $g_l$ as the mean values of the moments of a truncated Hadamard product over the zeros of the Riemann zeta function.

In the philosophy of Conjecture \ref{KSconjecture}, the parameter $N$ on the random matrix side corresponds to $\log T$ on the number theory side. Hence the leading order is the same in both settings, and the leading coefficient on the random matrix side predicts that on the number theory side, up to an arithmetic constant given in (\ref{arithmeticfactor}). This suggests that the connection (\ref{predicttheleadingcoefficient}) extends to the moments of higher order of derivatives of the Riemann zeta function on the critical line, or similarly to the Hardy $\mathsf{Z}$-function, defined by
\[
\mathsf{Z}(t) =  \pi^{-\textnormal{i}t/2} \frac{\Gamma(1/4 + \textnormal{i}t/2)}{|\Gamma(1/4+\textnormal{i}t/2)|}\zeta\!\left(\tfrac{1}{2}+\textnormal{i}t\right).
\]
It satisfies $|\mathsf{Z}(t)| = \big|\zeta(\tfrac{1}{2} + \textnormal{i}t)\big|$ and $\mathsf{Z}(t) \in \mathbb{R}$ for all $t \in \mathbb{R}$.

\begin{conjecture}\label{generalKSconjecture}
Let $n_{1}>n_{2}\geq 0$ and $l\geq 1,h\geq 0$ be integers with $h\leq l$, and let $c_{l}$ be the arithmetic factor as given in (\ref{arithmeticfactor}). Then as $T\rightarrow \infty$,
\begin{equation*}\label{generalpredicttheleadingcoefficient}
\frac{1}{T}\int_0^T \left|\zeta^{(n_{1})}\!\left(\tfrac{1}{2}+it\right)\right|^{2h}\left|\zeta^{(n_{2})}\!\left(\tfrac{1}{2}+it\right)\right|^{2l-2h} dt \sim c_l 
a_{h,l}(n_{1},n_{2})(\log T)^{m_{h,l}(n_{1},n_{2})},
\end{equation*}
and 
\begin{equation}\label{generalpredicttheleadingcoefficient2}
\frac{1}{T}\int_0^T \left|\mathsf{Z}^{(n_{1})}(t)\right|^{2h}\left|\mathsf{Z}^{(n_{2})}(t)\right|^{2l-2h} dt \sim c_l 
b_{h,l}(n_{1},n_{2})(\log T)^{m_{h,l}(n_{1},n_{2})}, \end{equation}
where $a_{h,l}(n_{1},n_{2})$, $b_{h,l}(n_{1},n_{2})$ and $m_{h,l}(n_{1},n_{2})$ appear in the leading terms of the asymptotic formulas
\begin{equation}\label{generalasyforcue}
\int_{U(N)} |\Lambda_N^{(n_{1})}(1)|^{2h}|\Lambda_N^{(n_{2})}(1)|^{2l-2h} \, d\mu_{N} \sim a_{h,l}(n_{1},n_{2})N^{m_{h,l}(n_{1},n_{2})}, \quad N\to \infty,
\end{equation}
and 
\begin{equation}\label{generalasyforcue2}
\int_{U(N)} |Z_N^{(n_{1})}(1)|^{2h}|Z_N^{(n_{2})}(1)|^{2l-2h} \, d\mu_{N} \sim b_{h,l}(n_{1},n_{2})N^{m_{h,l}(n_{1},n_{2})}, \quad N\to \infty.
\end{equation}
\end{conjecture}
Here $f^{(n)}(x)$ denotes $\tfrac{d^{n}f}{dx^{n}}(x)$, and $Z_{N}(s)$ is the analogue of Hardy's $\mathsf{Z}$-function, defined by
\begin{equation}
Z_{N}(s) := e^{-\pi i N/2}e^{i \sum_{n=1}^{N}\theta_{n}/2}
s^{- N/2}\Lambda_{N}(s).
\end{equation}

 For $n_{2}=0$ and $h=0$, Conjecture \ref{generalKSconjecture} reduces to the Keating–Snaith conjecture. For $n_{1}=1,n_{2}=0$, it was proposed in Hughes’ thesis \cite{Hu01}, where he also proved that $m_{h,l}(1,0)=l^2+2h$ and gave explicit formulas for $a_{h,l}(1,0)$ in (\ref{generalasyforcue}) and $b_{h,l}(1,0)$ in (\ref{generalasyforcue2}). We also refer, for example, to \cite{ABGS21}, \cite{AKW22},  \cite{B+19}, \cite{BW25}, \cite{CRS06}, \cite{De08}, \cite{FW06}, \cite{SW24} and \cite{Wi12} for alternative expressions for for $a_{h,l}(1,0)$ and $b_{h,l}(1,0)$. For general $n \geq 2$, Keating and Wei \cite{KW24a} showed that $m_{h,l}(n_{1},n_{2})=l^2+2hn_{1}+2(l-h)n_{2}$, and expressed $a_{h,l}(n_{1},n_{2})$ and $b_{h,l}(n_{1},n_{2})$ in terms of certain partition sums. 
It is proved that both $a_{h,l}(n_{1},n_{2})$ in (\ref{generalasyforcue}) and $b_{h,l}(n_{1},n_{2})$ in (\ref{generalasyforcue2}) can be expressed in terms of derivatives of
the Hankel determinant
\begin{equation}\label{bessel}
\det \Big [ I_{j+k+1}(2\sqrt{x}) \Big ]_{j,k=0,\dots,l-1}
\end{equation}
evaluated at $0$; see \cite{CRS06} for $n_{1}=1$, $n_{2}=0$ and \cite{KW24a} for general $n_{1}$, $n_{2}$. A relationship to a $\sigma$-Painlev\'e III$'$ transcendent and thus an integrable system was used in \cite{FW06} to systematically compute $\{a_{h,l}(1,0)\}, \{b_{h,l}(1,0)\}$, thereby extending the $n_{1}=1$, $n_{2}=0$ tabulations of
\cite{CRS06}. However, it was observed in \cite{KW24b} that beyond the order of the power series of (\ref{bessel}) needed for the $(n_{1},n_{2})=(1,0)$ application, the $\sigma$-Painlev\'e III$'$ equation does not possess a unique solution. Moreover, in \cite{KW24b}, to circumvent this obstruction an alternative integral structure was isolated, namely that of a higher order linear differential equation. To put this alternative structure in a wider context 
is the main purpose of our paper \cite{FW25}.

Crucial for this purpose is the rewrite of the Hankel determinant (\ref{bessel}) as a unitary matrix integral, namely,
\begin{equation}\label{unitaryintegral}
\Big \langle (\det U)^{l} e^{z {\rm Tr}(U + U^\dagger)} \Big \rangle_{U(l)} =
(-1)^{l(l-1)/2} \det \Big [ I_{j+k+1}(2z) \Big ]_{j,k=0,\dots,l-1}.
\end{equation}
The matrix integral written in terms of its eigenvalues can be identified, via a limiting procedure, as belonging to the Selberg correlation integral class with a single independent variable. In this setting beginning with the pioneering work of Davis \cite{Da68}, and further developed in the work of Forrester and Rains \cite{FR12}, there is a systematic way to deduce an $(l+1) \times (l+1)$ matrix differential equation, from which a higher order scalar differential equation follows. Moreover, this strategy applies equally to the generalization of the matrix integral as being over the circular $\beta$ ensemble, and with the power of the factor of the determinant a free parameter.

As an application distinct from the number theory motivation, it was noted from
 \cite{Ge90}, \cite{Ra98} that 
 \begin{equation}\label{l=0}
\Big \langle e^{z {\rm Tr} (U + U^\dagger)} \Big \rangle_{U(l)} =
1 + \sum_{N=1}^\infty \frac{T_l(N)}{(N!)^2} z^{2N},
\end{equation}
where $T_{l}(N)$ counts the number of permutations of $\{1, 2, \ldots, N \}$ whose longest increasing subsequence has length at most $l$. Our results then provide an alternative method to that of \cite{BG00} for computing the scalar linear differential equation of (\ref{l=0}) of degree $l+1$ for $l=2,3,\ldots,7$, and extend this computation to general larger values of $l$. Furthermore, we obtain an efficient recursive algorithm to compute the coefficients of $T_{l}(N)$ for arbitrary $N$ and $l$.

In the number theory context of
(\ref{unitaryintegral}), following \cite{KW24b}
we use the matrix linear differential formulation to deduce a recursive relation for the coefficients of the Taylor expansion of (\ref{bessel}) at $x=0$ for any $l$. This allows for the efficient computation of $a_{h,l}(n_{1},n_{2})$ and $b_{h,l}(n_{1},n_{2})$ involved in Conjecture \ref{generalKSconjecture} for large $h,l$ and any fixed $n_{1},n_{2}$.
The availability of such data for $a_{h,l}(n_{1},n_{2})$ and $b_{h,l}(n_{1},n_{2})$ is not only of computational interest, but also useful in problems in number theory concerning large gaps between zeros of higher-order derivatives of Hardy’s $\mathsf{Z}$-function. 

To illustrate this, let
\[
\Theta^{(m)} = \limsup_{n\rightarrow \infty} \frac{t_{n+1}^{(m)} - t_n^{(m)}}{2\pi / \log t_n^{(m)}},
\] 
where $\{t_n^{(m)}\}$ is the sequence of non-negative zeros of $\mathsf{Z}^{(m)}(t)$, counted according to multiplicity.
For $m=0$, the quantity $\Theta^{(0)}$ was introduced by Hall \cite{Hall99}. Under the Riemann hypothesis, 
\[
\Theta^{(0)} = \Theta:= \limsup_{n\rightarrow \infty} \frac{\gamma_{n+1} - \gamma_n }{2\pi / \log \gamma_n},
\]
where $\{\gamma_n\}$ is the non-decreasing sequence of imaginary parts of the zeros of the Riemann zeta function in the upper half-plane, counted with multiplicity. The quantity $\Theta$ was introduced by Selberg \cite{S46}, who proved that $\Theta > 1$.  
Note that $\Theta$ and $\Theta^{(m)}$ have an average value of 1 under the Riemann hypothesis. This has motivated number theorists to investigate how large these quantities can be.
Conrey, Ghosh, and Gonek showed that $\Theta > 2.337$ under the Riemann hypothesis \cite{CGG84} and $\Theta> 2.68$ under the Generalized Riemann hypothesis \cite{CGG86}. In \cite{Hall04}, Hall did not assume the Riemann hypothesis, but instead assumed the truth of Conjecture \ref{generalKSconjecture} for $n_{1}=1$, $n_{2}=0$, $l=6$, $0\leq h\leq 6$, and used the explicit values of $b_{h,6}(1,0)$, then proved that $\Theta^{(0)} > 4.29$ times the average gap length. This demonstrates the usefulness of having precise data for $a_{h,l}(n_{1},n_{2})$ and $b_{h,l}(n_{1},n_{2})$. 

It is widely conjectured that $\Theta^{(0)} = \infty$, though this remains open even under GRH. While no direct connection is known between the Riemann hypothesis and Conjecture \ref{generalKSconjecture}, the availability of explicit data for $b_{h,l}(n_{1},n_{2})$ provides a potential way to push lower bounds further. Indeed, recent work of Bui and Hall \cite{BH23} obtained new lower bounds for $\Theta^{(m)}$ for certain $m \geq 1$. This suggests that extending Hall’s method in \cite{Hall04}, together with data for $b_{h,l}(n_{1},n_{2})$ in (\ref{generalasyforcue2}), could yield sharper results, assuming Conjecture \ref{generalKSconjecture} for these $b_{h,l}(n_{1},n_{2})$. In what follows, we establish two new results illustrating how the above suggestion can be implemented.
\begin{theorem}\label{largegap1}
Assume that the asymptotic formula \eqref{generalpredicttheleadingcoefficient2}
in Conjecture \ref{generalKSconjecture} holds for $n_{1}=2$, $n_{2}=1$, and for
all integers $h$ with $0 \leq h \leq 4$. Then
\begin{equation}
\Theta^{(1)} > 1.98.
\end{equation}
\end{theorem}

\begin{theorem}\label{largegap2}
Assume that the asymptotic formula \eqref{generalpredicttheleadingcoefficient2}
in Conjecture \ref{generalKSconjecture} holds for $n_{1}=3$, $n_{2}=2$, and for
all integers $h$ with $0 \leq h \leq 4$. Then
\begin{equation}
\Theta^{(2)} > 1.71.
\end{equation}
\end{theorem}

As far as we know, the two lower bounds provided in Theorems~\ref{largegap1} and~\ref{largegap2} are sharper than the previously known bounds. Conrey and Ghosh~\cite{CG} showed that $\Theta^{(1)} > 1.4$ under the Riemann Hypothesis, while Hall later obtained the unconditional bound $\Theta^{(1)} > 1.5462$~\cite{Hall04}. More recently, Bui and Hall~\cite{BH23} proved unconditionally that $\Theta^{(1)} > 1.9$ and $\Theta^{(2)} > 1.606$. Under the assumptions of Theorems~\ref{largegap1} and~\ref{largegap2}, we obtain strictly larger lower bounds than all of these earlier results.

\section{Proofs of Theorems \ref{largegap1} and \ref{largegap2}}

We start with the generalized Wirtinger inequality given in \cite[Theorem 2]{Hall02}.

\begin{lemma}\label{inequality}
Suppose that $l \in \mathbb{N}$, $y(x) \in C^{2}[0,\pi]$, and
\[
y(0)=y(\pi)=0.
\]
Let $H(u)$ be an even function, increasing and strictly convex on $\mathbb{R}_{+}$, such that
\[
H(0)=H'(0)=0,
\qquad
\text{and } \; u H''(u) \to 0 \quad \text{as } u \to 0.
\]
Then
\begin{equation}\label{eq:GW}
\int_{0}^{\pi} H\!\left(\frac{y'(x)}{y(x)}\right) y(x)^{2l}\, dx
\;\ge\;
(2l-1)\lambda
\int_{0}^{\pi} y(x)^{2l}\, dx,
\end{equation}
where $\lambda=\lambda(l,H)$ is determined by the equation
\begin{equation}\label{eq:lambda}
\int_{0}^{\infty}
\frac{G'(u)}{G(u)+(2l-1)\lambda}\,
\frac{du}{u}
= l\pi,
\qquad
G(u):=uH'(u)-H(u).
\end{equation}
In particular, in the case
\begin{equation}\label{eq:Hspecial}
H(u):=\sum_{h=1}^{l}
\frac{2l-1}{2h-1}
\binom{l}{h} v_h u^{2h},
\qquad
v_h \ge 0,\quad v_l=1,
\end{equation}
equation \eqref{eq:lambda} becomes
\begin{equation}\label{eq:speciallambda}
\int_{-\infty}^{\infty}
\frac{
u^{2l-2}
+ \binom{l-1}{l-2} v_{l-1} u^{2l-4}
+ \binom{l-1}{l-3} v_{l-2} u^{2l-6}
+ \cdots + v_1
}{
u^{2l}
+ \binom{l}{l-1} v_{l-1} u^{2l-2}
+ \cdots + \binom{l}{1} v_1 u^{2}
+ \lambda
}
\, du
= \pi.
\end{equation}
We may rewrite \eqref{eq:lambda} in the form
\begin{equation}\label{eq:logform}
\int_{0}^{\infty}
\log\!\left(1+\frac{G(u)}{(2l-1)\lambda}\right)
\frac{du}{u^{2}}
= l\pi.
\end{equation}
From this it is evident that $\lambda$ is an increasing function of $G$
(in the obvious sense).
Thus, in the special case \eqref{eq:Hspecial},
$
\lambda=\lambda(v_1,v_2,\dots,v_{l-1})
$
is an increasing function of each $v_j$.
The inequality is sharp.
\end{lemma}

Using the above inequality~\eqref{eq:GW}, Hall~\cite{Hall02,Hall04} connects the lower bound for $\Theta^{(0)}$ to the optimization problem in equation~\eqref{Xeq} (with $n=0$) below, aiming to maximize its solution by choosing appropriate parameters.
The following result is a generalization of \cite[Proposition~A]{Hall04}, extending the case $B(h,l;0)$ (defined in \eqref{Rdef} below) to $B(h,l;n)$ for any integer $n \ge 0$.

\begin{lemma}\label{optimization}
Let $l\in \mathbb{N}$. Suppose that
\begin{equation}\label{Gdef}
G(u) = (2l-1)\sum_{h=1}^{l} \binom{l}{h} v_h u^{2h}, 
\qquad v_h \in \mathbb{R}, \quad v_l = 1,
\end{equation}
is strictly increasing on $\mathbb{R}_{+}$, and let 
$\lambda = \lambda(v_1, v_2,\ldots, v_{l-1})$ be the unique root of the equation
\begin{equation}\label{leq}
\int_{0}^{\infty} \log\!\left(1 + \frac{G(u)}{(2l-1)\lambda}\right)\frac{du}{u^{2}} = l\pi.
\end{equation}
Let $\mathcal{X}$ be the real positive root of the equation
\begin{equation}\label{Xeq}
\sum_{h=1}^{l} \frac{2l-1}{2h-1} \binom{l}{h} B(h,l;n)\, v_h \mathcal{X}^{h}
= (2l-1)\lambda(v_1, v_2,\ldots, v_{l-1}) B(0,l;n),
\end{equation}
where for integer $n\geq 0$, $h=0,\ldots,l$,
\begin{equation}\label{Rdef}
B(h,l;n) := 4^{\,h-l}\frac{d_{h,l}(n+1,n)}{d_{l,l}(n+1,n)},
\end{equation}
with $d_{h,l}(n+1,n)$ appearing in the following asymptotic formula,
\begin{equation}\label{defofd}
\frac{1}{T}\int_0^T \left|\mathsf{Z}^{(n+1)}(t)\right|^{2h}\left|\mathsf{Z}^{(n)}(t)\right|^{2l-2h} dt \sim
d_{h,l}(n+1,n)(\log T)^{l^2+2ln+2h},\,T\rightarrow \infty.
\end{equation}
Assume that the parameters $v_1, v_2, \ldots, v_{l-1}$ are chosen so as to maximise $\mathcal{X}$.
Then this maximal value is no more than $(\Theta^{(n)})^2$.
\end{lemma}

\begin{proof}
The following argument is similar to that of \cite[Section 2]{Hall02}. To make the paper self-contained, we include the proof below.
Let $n\geq 0$.
Denote all the zeros of the function $\mathsf{Z}^{(n)}(t)$ in the interval $[T,2T]$ by
$
x_{1} \le x_2 \le \cdots \le x_{m}
$, that is $x_{1}$ is the first zero not less than $T$ and $x_m$ the last
zero not exceeding $2T$.
Suppose that $1\le j < m$, we have
\begin{equation}\label{eq:90}
x_{j+1} - x_j\le \frac{2\pi\kappa }{\log T}.
\end{equation}
By a linear transformation from $[0,\pi]$ to $[a,b]$, we apply Lemma~\ref{inequality}, with $H$ as in \eqref{eq:Hspecial} and $y(x)=\mathsf{Z}^{(n)}(x)$, to obtain 
\begin{align}\label{eq:ineq32}
&\int_{x_{j}}^{x_{j+1}}
\Bigg(
\sum_{h=1}^{l}
\frac{2l-1}{2h-1}
\binom{l}{h}
\left(\frac{x_{j+1}-x_{j}}{\pi}\right)^{2h}
v_h
\Big(\mathsf{Z}^{(n)}(t)\big)^{2l-2h} \Big(\mathsf{Z}^{(n+1)}(t)\Big)^{2h}\nonumber\\
&-
(2l-1)\lambda(v_1,\dots,v_{l-1}) \mathsf{Z}^{(n)}(t)^{2l}
\Bigg)
dt\ge 0 .
\end{align}
Since $G(u)$ is strictly increasing on $\mathbb{R}_{+}$, the inequality remains valid if we replace
$
\frac{x_{j+1}-x_j}{\pi}
$
by
$
\frac{2\kappa}{\log T}
$
throughout. Note that $\mathsf{Z}(t)$ is a real-valued function. Hence, for integers $l,h$, we have $(\mathsf{Z}^{(n)}(t))^{2(l-h)}=|\mathsf{Z}^{(n)}(t)|^{2(l-h)}$. Using \eqref{defofd} and summing the resulting inequality (\ref{eq:ineq32}) over $j$, we may, up to a negligible error, replace the limits of integration $x_1$ and $x_m$ by $T$ and $2T$, respectively. We then obtain
\begin{align*}\label{eq:33}
&T (\log T)^{^{l^2+2nl}}
\Bigg(
\sum_{h=1}^{l}
\frac{2l-1}{2h-1}
\binom{l}{h}
 (2\kappa)^{2h}d_{h,l}(n+1,n)v_h\\
&-
(2l-1)\lambda(v_1,\dots,v_{l-1}) d_{0,l}(n+1,n)
\Bigg)
\geq  o\!\left(T (\log T)^{l^2+2ln}\right) .
\end{align*}
Consequently,
\begin{equation}\label{eq:34}
\kappa^2 \ge \mathcal{X} + o(1),
\quad (T \to \infty),\quad (\Theta^{(n)})^2\geq \mathcal{X},
\end{equation}
where $\mathcal{X}$ is the real positive real root of
\begin{equation}\label{eq:Xeq}
\sum_{h=1}^{l}
\frac{2l-1}{2h-1}
\binom{l}{h}
B(h,l;n) v_h \mathcal{X}^{h}
-
(2l-1)\lambda(v_1,\dots,v_{l-1}) B(0,l;n)
= 0 .
\end{equation}
\end{proof}

In the proofs of Theorems~\ref{largegap1} and~\ref{largegap2}, it is important to have the following data for the leading coefficients of the joint moments of the derivatives of the analogue Hardy $\mathsf{Z}$-function in random matrix theory.
\begin{lemma}\label{Table}
We have the following values of $b_{h,4}(n+1,n)$ in Table~\ref{tab:bhln} for $n = 1,2$ and $h = 0,\ldots,4$. Here $b_{h,l}(n_{1},n_{2})$ is defined as in (\ref{generalasyforcue2}).
\begin{table}[h]
\renewcommand{\arraystretch}{1.6}
\setlength{\tabcolsep}{18pt}
\centering
\caption{Values of $b_{h,4}(n+1,n)$ for $n = 1,2$ and $h = 0,\ldots,4$.}
\label{tab:bhln}
\begin{tabular}{|c|c|c|}
\hline
               & $n=1$ & $n=2$ \\ \hline
$b_{0,4}(n+1,n)$ & $ \frac{31}{2^{20}\cdot 3^{10}\cdot 5^4 \cdot 7^2 \cdot 11 \cdot 13  }$ & $ \frac{103\cdot 413129}{2^{28} \cdot 3^{12} \cdot 5^{5} \cdot 7^3 \cdot 11^2 \cdot 13^2 \cdot 17 \cdot 19 \cdot 23}$ \\ \hline
$b_{1,4}(n+1,n) $& $\frac{71}{2^{22}\cdot 3^{9}\cdot 5^4 \cdot 7^2 \cdot 11^2 \cdot 13 \cdot 17  }$ & $ \frac{58452853}{2^{30} \cdot 3^{11} \cdot 5^{7} \cdot 7^4 \cdot 11^3 \cdot 13^2 \cdot 17 \cdot 19  } $ \\ \hline
$b_{2,4}(n+1,n) $& $\frac{43\cdot 3467}{2^{24}\cdot 3^{10}\cdot 5^5 \cdot 7^4 \cdot 11^2 \cdot 13 \cdot 17\cdot 19 }$ & $\frac{67\cdot 85489981}{2^{32} \cdot 3^{12} \cdot 5^{7} \cdot 7^5 \cdot 11^3 \cdot 13^2 \cdot 17 \cdot 19 \cdot 23 }$ \\ \hline
$b_{3,4}(n+1,n)$ & $\frac{271\cdot 11483}{2^{26}\cdot 3^{13}\cdot 5^5 \cdot 7^3 \cdot 11^2 \cdot 13^2 \cdot 17\cdot 19 }$ & $\frac{71\cdot 389\cdot 139091}{2^{34} \cdot 3^{8} \cdot 5^{6} \cdot 7^5 \cdot 11^3 \cdot 13^2 \cdot 17^2 \cdot 19 \cdot 23 \cdot 29  }$ \\ \hline
$b_{4,4}(n+1,n)$ & $\frac{103\cdot 413129}{2^{28}\cdot 3^{12}\cdot 5^5 \cdot 7^3 \cdot 11^2 \cdot 13^2 \cdot 17\cdot 19\cdot 23 }$ & $\frac{449\cdot 1721279377}{2^{36} \cdot 3^{10} \cdot 5^{7} \cdot 7^{4} \cdot 11^3 \cdot 13^3 \cdot 17^2 \cdot 19 \cdot 23   \cdot 31 }$ \\ \hline
\end{tabular}
\end{table}
\end{lemma}
\begin{proof}
By \cite[Theorem 24]{KW24a} with $n_{1}=2,n_{2}=1$, we have, for $h=0,\ldots,l$,
\begin{align*}
&b_{h,l}(2,1)
\\
=&(-1)^{l-h+\frac{l(l-1)}{2}}
\sum_{k=0}^{2h}
\binom{2h}{k}
\sum_{\substack{h_{1}+\cdots+h_{l}=k\\m_{0}+\ldots+m_{l}=2l+2h-2k}}(-\frac{1}{2})^{m_{0}}
\binom{k}{h_{1},\ldots,h_{l}}\\
&\binom{2l+2h-2k}{m_{0},\ldots,m_{l}}\prod_{j=1}^{l}\frac{1}{(2l+m_{j}+2h_{j}-j)!}\prod_{1\leq j<s\leq l}\left(m_{s}+2h_{s}-m_{j}-2h_{j}-s+j\right).
\end{align*}
We set $l=4$ in the above computation performed using Maple.
Then we obtain the values of $b_{h,4}(2,1)$ for $h=0,\ldots,4$.  
Similarly, by applying \cite[Theorem 24]{KW24a} with $n_1=3$, $n_2=2$, and $l=4$, we obtain the values of $b_{h,4}(3,2)$ for $h=0,\ldots,4$.
\end{proof}

Methods for locating large positive real solutions of equation (\ref{Xeq}) with suitable parameters were studied in \cite{Hall04}. By combining these methods with the random matrix data provided in Lemma \ref{Table}, we are now ready to prove Theorems \ref{largegap1} and \ref{largegap2}.

\begin{proof}[Proof of Theorem \ref{largegap1}]
By the assumption of this theorem, for $h=0,\ldots,4$ we have
\[
d_{h,4}(2,1)=c_{4}\, b_{h,4}(2,1),
\]
where $d_{h,4}$ and $b_{h,4}$ are defined in \eqref{defofd} and
\eqref{generalasyforcue2}, respectively. Here $c_{4}$ is the
arithmetic factor given in \eqref{arithmeticfactor}. For $n\geq 0$ and $h=1,\ldots,4$, define
\begin{equation}\label{defofA}
A(h,4;n)=\left|\frac{2h-1}{2h-3}\right|\frac{B(h-1,4;n)}{B(h,4;n)},
\end{equation}
where $B(h,4;n)$ is defined in \eqref{Rdef}. By Lemma \ref{Table}, we have
\[
\begin{aligned}
&A(1,4;1)=\frac{5797}{213}, \quad
A(2,4;1)=\frac{2974545}{149081},\\
& A(3,4;1)=\frac{87212385}{21783251}, \quad
A(4,4;1)=\frac{501014773}{638284305}.
\end{aligned}
\]
Set
\begin{equation}\label{abcd}
A:=A(1,4;1), \quad B:=A(2,4;1), \quad C:=A(3,4;1), \quad D:=A(4,4;1).
\end{equation}
Then, by an argument almost identical to that in \cite[Section~7]{Hall04},
with the slight difference that the quantities $A,B,C,D,X$ there are
replaced by \eqref{abcd} and $\mathcal{X}$ here, respectively, we find that
\[
\mathcal{X}
=
\frac{AB(C-D)+C(B^2+CD-2BD)}{16(B-C)(C-D)}
=
3.93116\ldots
\]
is a solution of equation \eqref{Xeq}. In this case, the corresponding parameters $v_{1},v_{2},v_{3}$ and
$\lambda$ in \eqref{Xeq} are computed in the same manner as in
\cite[Section~7]{Hall04}, as we now explain. Define
\[
\tau_1=\frac{4C-4D}{D(B-C)}, \quad
\tau_2=\frac{B-D}{D(B-C)}, \quad
\tau_3=\frac{4C-4D}{CD(B-C)}, \quad
\tau_4=\frac{C-D}{CD(B-C)},
\]
and
\[
\sigma_1=\tau_{1}\mathcal{X}, \quad
\sigma_2=\tau_{2}\mathcal{X}, \quad
\sigma_3=\tau_{3}\mathcal{X}^2, \quad
\sigma_4=\tau_{4}\mathcal{X}^2.
\]
By comparing coefficients on both sides of the identity
\[
t^8 - 4 v_3 t^6 + 6 v_2 t^4 - 4 v_1 t^2 + \lambda
=
\bigl(t^4 + \sigma_2 t^2 + \sigma_4\bigr)^2
-
\bigl(\sigma_1 t^3 + \sigma_3 t\bigr)^2,
\tag{81}
\]
we obtain
\[
v_{1}=0.95297\ldots,\quad
v_{2}=1.01903\ldots,\quad
v_{3}=1.07681\ldots,\quad
\lambda=0.98492\ldots
\]
Note that $v_{1},v_{2},v_{3}$ are all positive. This ensures that $G(u)$ is
strictly increasing on $\mathbb{R}_{+}$, as required by the assumptions of
Lemma~\ref{optimization}. Therefore, by Lemma~\ref{optimization} with $l=4$ and $n=1$,
\[
\Theta^{(1)} \geq \sqrt{\mathcal{X}} > 1.98,
\]
as claimed in Theorem~\ref{largegap1}.
\end{proof}

\begin{proof}[Proof of Theorem \ref{largegap2}]
According to the assumptions of the theorem, for each $h = 0,\ldots,4$, we have
\[
d_{h,4}(3,2) = c_4 \, b_{h,4}(3,2),
\]
where $d_{h,4}$ and $b_{h,4}$ are as defined in \eqref{defofd} and \eqref{generalasyforcue2}, respectively, and $c_4$ is the arithmetic factor as in \eqref{arithmeticfactor}.
Let $A(h,4;2)$ be as in (\ref{defofA}). Then by Lemma \ref{Table}, we have 
\[
\begin{aligned}
&A(1,4;2)=\frac{81913152475}{4033246857}, \quad
A(2,4;2)=\frac{84698183997}{5727828727},\\
& A(3,4;2)=\frac{2823819562411}{933497701947}, \quad
A(4,4;2)=\frac{13933317551283}{22412778767917}.
\end{aligned}
\]
Then, by an argument almost identical to the proof of Theorem~\ref{largegap1}, with the only difference being that we replace 
$A(1,4;1)$, $A(2,4;1)$, $A(3,4;1)$, $A(4,4;1)$ by 
$A(1,4;2)$, $A(2,4;2)$, $A(3,4;2)$, $A(4,4;2)$, we obtain
\[
\mathcal{X} = 2.94783\ldots
\]
Correspondingly, 
\[
v_1 = 0.86987\ldots, \quad v_2 = 0.87121\ldots, \quad v_3 = 0.89954\ldots, \quad \lambda = 0.89144\ldots.
\]
Since $v_1, v_2, v_3$ are all positive, $G(u)$ is strictly increasing on $\mathbb{R}_+$. This verifies the requirement in the assumptions of Lemma~\ref{optimization}. Therefore, by Lemma~\ref{optimization} with $l=4$ and $n=2$, we have
\[
\Theta^{(2)} \geq \sqrt{\mathcal{X}} > 1.71,
\]
as claimed in Theorem~\ref{largegap2}.
\end{proof}
\begin{remark}
The lower bounds for $\Theta^{(1)}$ and $\Theta^{(2)}$ in Theorems~\ref{largegap1} and~\ref{largegap2}, and more generally for $\Theta^{(n)}$, can be improved under the assumption of Conjecture~\ref{generalKSconjecture} for larger moment exponents~$l$. In particular, for $l=5,6$, optimization algorithms for maximizing $\mathcal{X}$ in equation~\eqref{Xeq} are provided in \cite[Section~8]{Hall04} and \cite[Appendix]{Hall04}.
 Based on these results, one can apply arguments similar to those used in Theorems~\ref{largegap1} and~\ref{largegap2} with $l=5,6$ to further improve the lower bounds for $\Theta^{(n)}$. It is worth mentioning that this approach requires the computation of $b_{h,l}(n+1,n)$ for $l=5,6$. These quantities can be obtained either from \cite[Theorem~24]{KW24a}, as used in Lemma~\ref{Table}, or from \cite[Theorem~3]{KW24b} and \cite[Theorem~1.11]{AGKW24}, together with the recursive formulas for the Taylor coefficients of (\ref{bessel}) at $x=0$ (see \cite[Theorem~5]{KW24b} and \cite[Proposition~4.5]{FW25}), which allow for more efficient computation.
 We do not pursue this direction further in the present paper, since Theorems~\ref{largegap1} and~\ref{largegap2} already achieve the main objectives of this work, and the above mentioned extensions follow from essentially the same method.
\end{remark}

\begin{acknowledgement}
The work of PJF was supported by the Australian Research Council Discovery Project
DP250102552.
The work of FW was supported by the Royal Society, grant URF$\backslash$R$\backslash$231028. Both authors thank the organisers of the MATRIX program ``Log-gases in Caeli Australi'', held in Creswick, Victoria, Australia during the first half of August 2025 
for the stimulating event and 
facilitating our collaboration.
\end{acknowledgement}

\bigskip

\providecommand{\bysame}{\leavevmode\hbox to3em{\hrulefill}\thinspace}
\providecommand{\MR}{\relax\ifhmode\unskip\space\fi MR }
\providecommand{\MRhref}[2]{%
  \href{http://www.ams.org/mathscinet-getitem?mr=#1}{#2}
}
\providecommand{\href}[2]{#2}

\end{document}